\documentclass[pra,twocolumn]{revtex4-1}

\usepackage{amsmath,amssymb,amsthm,easybmat}
\usepackage{verbatim}
\usepackage{enumerate}
\usepackage{color}
\usepackage{bm}
\usepackage{graphicx}
\usepackage{longtable}
\usepackage{hyperref}
\usepackage{array}
\newcolumntype{x}[1]{%
>{\centering\arraybackslash}p{#1}}%

\newtheorem{proposition}{Proposition}
\newtheorem{theorem}{Theorem}

\newcommand{\ket}[1]{|#1\rangle}

\newcommand{\mc}[1]{\mathcal{#1}}
\newcommand{\mbf}[1]{\mathbf{#1}}
\begin{document}
\title{ Bell Inequalities with Communication Assistance}

\author{Katherine Maxwell $^1$}
\email{kam4756@truman.edu}
\author{Eric Chitambar $^2$}
\email{echitamb@siu.edu}

\affiliation{$^1$ Truman State University, Kirksville, Missouri 63501, USA\\
$^2$ Department of Physics and Astronomy{\mbox ,} Southern Illinois University, 
Carbondale, Illinois 62901, USA}

\date{\today}

\begin{abstract}
In this paper we consider the possible correlations between two parties using local machines and shared randomness with an additional amount of classical communication.  This is a continuation of the work initiated by Bacon and Toner in Ref. [\textit{Phys. Rev. Lett.} \textbf{90}, 157904 (2003)] who characterized the correlation polytope for $2\times 2$ measurement settings with binary outcomes plus one bit of communication.  Here, we derive a complete set of Bell Inequalities for $3\times 2$ measurement settings and a shared bit of communication.  When the communication direction is fixed, nine Bell Inequalities characterize the correlation polytope, whereas when the communication direction is bi-directional, 143 inequalities describe the correlations.  We then prove a tight lower bound on the amount of communication needed to simulate all no-signaling correlations for a given number of measurement settings. 
\end{abstract}
\maketitle

\section{Introduction}

Bell Inequalities provide one way to draw a boundary between the quantum and classical regimes.  While they do not tell the whole ``quantum versus classical'' story, Bell Inequalities nevertheless allow us to definitively certify the existence of certain non-classical phenomena and reflect on their philosophical implications.  

To understand exactly why Bell Inequalities are such a fundamental concept in physics, it is perhaps easiest to consider a theoretical scenario involving two distant parties called Alice and Bob.  Two ``black boxes'' labeled $\mc{M}_A$ and $\mc{M}_B$ are distributed to Alice and Bob respectively (see Fig. \ref{Fig:Single_Box}).  No assumptions are made about how these boxes are built or what physical devices are inside.  The only known properties of these boxes is that $\mc{M}_A$ accepts an input number $i$ chosen from the set $\{0,1,...,M_A-1\}$ and outputs a number $a$ from the set $\{0,1,...,K_A-1\}$.  We think of $M_A$ as the number of measurement settings that Alice has for her device and $K_A$ as the number of measurement outcomes.  Bob's box behaves analogously.  What interests us are the outputs $a$ and $b$ that Alice and Bob obtain, respectively, given a certain choice of inputs $i$ and $j$.  In general this involves a probabilistic transition governed by the conditional probabilities $p(ab|ij)$, and each complete set of such $p(ab|ij)$ form a $K_AK_B\times M_AM_B$ stochastic matrix that describes the \textit{correlations} between Alice and Bob's boxes.  

\begin{figure}[b]
\includegraphics[scale=0.6]{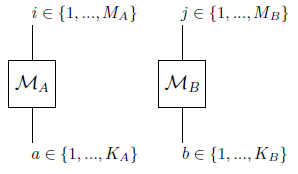}
\caption{\label{Fig:Single_Box} Alice and Bob have respective black boxes $\mc{M}_A$ and $\mc{M}_B$.  The input/output correlations are given by the conditional probabilities $p(ab|ij)$.} 
\end{figure} 

We now ask what correlations are possible given certain physical restrictions on the boxes $\mc{M}_A$ and $\mc{M}_B$.  Models consistent with classical physics consist of local boxes with shared randomness (LSR), which means that the probabilities $p(ab|ij)$ can be decomposed as 
\begin{equation}
\label{Eq:LSR}
p(ab|ij)=\sum_\lambda p(\lambda)q_A(a|i\lambda)q_B(b|j\lambda).
\end{equation}  
Here, $\lambda$ is some variable shared between Alice and Bob according to distribution $p(\lambda)$, and $q_A(a|i\lambda)$ (resp. $q_B(b|j\lambda)$) are conditional distributions that give a complete local description for the operation of $\mc{M}_A$ (resp. $\mc{M}_B$).  The essence of Bell's original paper \cite{Bell-1964a}, and further refined by Clauser, Holt, Shimony and Horne (CHSH) \cite{Clauser-1969a}, is that the correlations of any LSR boxes must satisfy certain inequalities that quantum boxes can break; hence quantum mechanics is able to generate nonlocal correlations.  By ``quantum boxes,'' we envision $\mc{M}_A$ and $\mc{M}_B$ as being two quantum systems prepared in some joint entangled state $\ket{\Psi}_{AB}$.  For $M_A=M_B=K_A=K_B=2$, there is essentially only one non-trivial inequality for LSR boxes which is appropriately referred to as the CHSH inequality \cite{Fine-1982a}.

Given that quantum boxes are more powerful than classical boxes, a natural question is what additional resources must be added to a classical model so that it can simulate quantum correlations.
From the practical perspective of experimentally simulating the correlations, one can ask \textit{how many bits of classical communication {\upshape (CC)} combined with local shared randomness are sufficient to reproduce quantum correlations} \cite{Maudlin-1992a, Brassard-1999a, Steiner-2000a}?  Bacon and Toner introduced the notion of ``Bell Inequalities with auxiliary communication'' which are generalized CHSH inequalities that identify all the correlations consistent with LSR and a stipulated amount of CC \cite{Bacon-2003a}.  They showed that, not surprisingly, for $M_A=M_B=K_A=K_B=2$, only one bit of communication is sufficient to replicate any possible quantum correlation.  This is because one bit of CC is capable of generating any set of statistics that is consistent with relativistic causation.  Assuming Alice and Bob's boxes to be spacelike separated, special relativity stipulates that Alice or Bob's choice of input cannot affect the output statistics of the other, a condition known as \textit{no-signaling} and characterized by: 
\begin{align}
\label{Eq:No-Signal-Alice}
q_A(a|i)&:=\sum_{b=0}^{K_B-1}p(ab|ij)\quad\forall a,i,j\\
\label{Eq:No-Signal-Bob}
q_B(b|j)&:=\sum_{a=0}^{K_A-1}p(ab|ij)\quad\forall b,i,j.
\end{align}
As quantum mechanics respects the no-signaling principle, 1 bit of CC suffices to simulate quantum statistics for $M_A=M_B=K_A=K_B=2$.  For larger number of measurement settings little is currently known about the amount of communication needed to simulate quantum correlations.  

The main contribution of this paper is a presentation of all the Bell Inequalities for $M_A=3$ and $M_B=K_A=K_B=2$ when assistance is provided by one bit of classical communication.  Note that when Bacon and Toner consider the $M_A=3$ case in Ref. \cite{Bacon-2003a}, they only compute expectation inequalities for observables with $\pm 1$ spectrum.  Computing the allowed probabilities $p(ab|ij)$ is a more general and complicated problem.  In fact, once beyond two measurement settings, the complexity of the CC-assisted problem increases quite dramatically, as our computations below will demonstrate.  Thus, any small step forward in understanding these correlations is valuable.  We give an overview of general CC-assisted LSR boxes in \ref{Sect:CC}.  We then present the Bell Inequalities in Section \ref{Sect:results} and prove lower bounds on the classical communication cost to simulate no-signaling correlations in Section \ref{Sect:LowerBounds}.

\section{A General CC-Assisted Framework}
\label{Sect:CC}

In this section we describe a general framework for using a limited amount of classical communication to generate correlations under a local shared randomness model.  Since we are interested in isolating the power of one CC bit on its own, we consider two different models: \textit{fixed-direction} communication and \textit{bi-directional} communication (see Fig. \ref{Fig:CC_Boxes}).  The former describes a scenario where Alice and Bob can only send CC in a certain direction whereas this directional restriction is removed in the latter.  Before describing these models in more detail, we first review the notion of locally equivalent correlations.

\begin{figure}[b]
\includegraphics[scale=0.5]{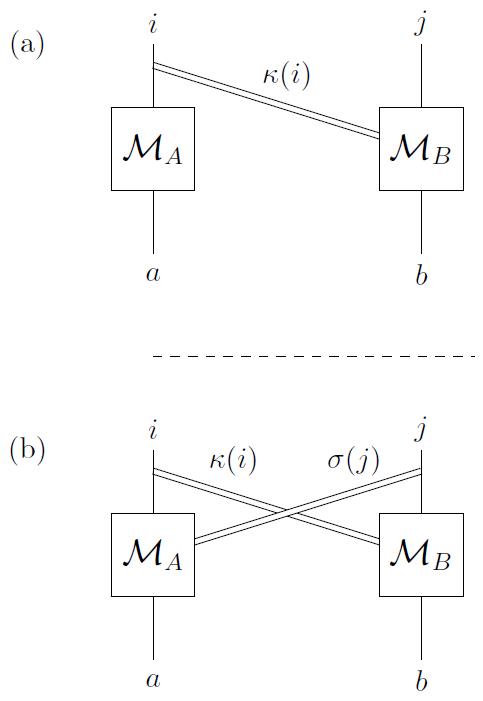}
\caption{\label{Fig:CC_Boxes} (a) An Alice $\to$ Bob fixed-direction communication scheme in which $\kappa(i)$ is sent from Alice to Bob. (b)  A bi-directional communication scheme sending data in both directions.  Note that model (b) contains both Alice $\to$ Bob and Bob $\to$ Alice fixed-direction communication schemes.} 
\end{figure} 

\subsection{Local Operations}

Even though two correlations $p(ab|ij)$ and $p'(ab|ij)$ may differ, they might be equivalent in their nonlocal content.  This would be the case if, for instance, $p'(ab|ij)=p(ab|\phi_A(i)j)$ $\forall a,b,i,j$, where $\pi_A:\{0,1,...,M_A-1\}\leftrightarrow\{0,1,...,M_A-1\}$ is some local permutation of Alice.  Clearly the correlations do not become any more or less nonlocal when Alice performs this permutation.  More generally, we say that one correlation matrix $p(ab|ij)$ is locally equivalent to another $p'(ab|ij)$ if (i) they are are related by an uncorrelated permutation of inputs: $p'(ab|ij)=p(ab|\pi_A(i)\pi_B(j))$, or (ii) they are related by an uncorrelated conditional permutation of outputs: $p'(ab|ij)=p(\pi_A^{(i)}(a)\pi_B^{(j)}(b)|ij)$, where $\pi_A^{(i)}$ is a particular choice of permutation for Alice depending on her input $i$, and likewise for $\pi_B^{(j)}$.

Based on Fig. \ref{Fig:CC_Boxes} and the CC-assisted model, it may seem that Alice/Bob may be able to relabel his/her output conditioned on the $\kappa(i)/\sigma(j)$ received from the other party.  In principle this is true.  However, it will be unnecessary to consider this dependence explicitly as long as we consider all possible encodings $\kappa(i)/\sigma(j)$ and all possible local maps of Alice and Bob; a conditional relabeling can be seen as just another encoding/decoding strategy.  This will be made more clear when we discuss the computation strategy below.

\subsection{Fixed-Direction Communication}

Without the classical communication, the correlations of a general LSR model satisfy Eq. \eqref{Eq:LSR}.  Suppose now that Alice is allowed to send $r$ bits of CC to Bob.  In general this can be represented by some $\lambda$-dependent function $\kappa_\lambda(i)$ that acts on Alice's input value $i$.   Hence, the resulting correlations satisfy $p(ab|ij)=\sum_\lambda p(\lambda)q_A(a|i\lambda)q_B(b|j\kappa_\lambda(i)\lambda)$.  We assume that each $\kappa_\lambda$ is a deterministic mapping since any unshared randomness of the functions $q_A(\cdot|i\lambda)$, $q_B(\cdot|j\kappa(i)\lambda)$, and $\kappa_\lambda(i)$ can be absorbed into the globally shared random variable $\lambda$.  Consequently, $p(ab|ij)$ can be expressed as a convex sum of deterministic strategies characterized by functions $\kappa(i)\in\{0,1,...,2^{r}-1\}$, $q_A(a|i)\in\{0,1,...,K_A-1\}$, and $q_B(b|j\kappa(i))\in\{0,1,...,K_B-1\}$ for all $i,j,a,b$.  As there are only a finite number of deterministic strategies, the correlations $p(ab|ij)$ generated by such strategies form the vertices of a convex polytope.  Following the standard procedure, we enumerate all polytope vertices and then generate their convex hull.

Let us restrict to $K_A=K_B=2$.  The correlation polytope will then have dimension $4M_AM_B$, but normalization can be enforced to eliminate probabilities of the form $p(11|ij)$ and reduce the dimension to $M_AM_B(K_AK_B-1)=3M_AM_B$.  The dimension can be reduced even further by noting that the no-signaling condition must hold from Bob to Alice.  This means that probabilities are restricted by Eq. \eqref{Eq:No-Signal-Alice}.  As a result, for each fixed $i$, we can eliminate $p(01|ij)$ in favor of $q_A(0|i)$, and in total the polytope will have dimension $M_A(2M_B+1)$ with the independent variables being $\{p(00|ij),p(10|ij),q_A(0|i)\}_{i=0;j=0}^{M_A-1;M_B-1}$.  To satisfy the LSR + CC constraint, the polytope vertices must satisfy
\begin{align}
\label{Eq:ProbPost-select}
p(00|ij)&=q_A(0|i)q_B(0|j\kappa(i))\notag\\
p(10|ij)&=[1-q_A(0|i)]q_B(0|j\kappa(i)).
\end{align}

We thus generate all the vertices by considering all possible functions $\kappa(i)$, $q_A(0|i)$, $q_B(0|j\kappa(i))$ and then forming the probabilities according to Eq. \eqref{Eq:ProbPost-select}.  The process can be simplified a bit since not all possibilities of $\kappa$ need to be considered.  First, it is the particular groupings $\kappa^{-1}(l)$ of $\{0,1,...,M_A-1\}$ that matter and not the values of $\kappa(i)$ themselves.  This is because we consider all possible mappings $q_B(0|jl)$ with $j\in\{0,1,...,M_B-1\}$ and $l\in\{0,1,...,2^r-1\}$.  Second, we can assume that $\kappa$ is surjective since any non-surjective map will generate probabilities that are also generated by a subjective one.  Thus, the total number of communication functions needing to be counted is equivalent to the number of ways $M_A$ elements can be grouped into exactly $2^r$ equivalent classes.  This is given by Stirling's Number of the Second Kind \cite{Abramowitz-2012a}: $\genfrac\{\}{0pt}{}{M_A}{2^r}$.

\subsection{Bi-Direction Communication}

Without the directional restriction, the $r$ bits can be split between Alice and Bob so that Alice sends $s$ bits and Bob sends $r-s$.  In general, the value of $s$ can depend on some variable $\lambda$, and we denote this dependence by $s_\lambda$.  The communication scheme can then be modeled by functions $\kappa_\lambda(i) \in \{0,1,\ldots,2^{s_\lambda}-1\}$ and $\sigma_\lambda(j) \in \{0,1,\ldots,2^{r-s_\lambda}-1\}$. The resulting correlations satisfy $p(ab|ij)=\sum_\lambda p(\lambda)q_A(a|\sigma_\lambda(j)i\lambda)q_B(b|j\kappa_\lambda(i)\lambda)$.  Like before, $p(ab|ij)$ can be expressed as a convex sum of deterministic strategies where $q_A(a|\sigma(j)i)\in\{0,1,...,K_A-1\}$ and $q_B(b|j\kappa(i))\in\{0,1,...,K_B-1\}$ for all $i,j,a,b$.

Restricting to $K_A=K_B=2$, the normalization constraint reduces the resulting polytope to dimension $3M_AM_B$. This will be the final dimension of the polytope since the no-signaling condition does hold in either direction. Therefore the independent variables are $\{p(00|ij),p(10|ij),q(01|ij)\}_{i=0;j=0}^{M_A-1;M_B-1}$.  To satisfy the LSR + CC constraint, the polytope vertices must satisfy
\begin{align}
\label{Eq:ProbPost-select2}
p(00|ij)&=q_A(0|i\sigma(j))q_B(0|j\kappa(i))\notag\\
p(10|ij)&=[1-q_A(0|i\sigma(j))]q_B(0|j\kappa(i))\notag\\
p(01|ij)&=q_A(0|i\sigma(j))[1-q_B(0|j\kappa(i))].
\end{align}

We generate the vertices by considering all possible functions $\kappa(i)$, $\sigma(j)$, $q_A(0|i\sigma(j))$, $q_B(0|j\kappa(i))$ for all values of $0<s\leq r$. For $r=1$, we have $s \in \{0,1\}$ corresponding respectively to Bob sending Alice one bit and Alice sending Bob one bit. Thus in this case, the vertices of the random-direction polytope will be those of the fixed-direction polytope with communication from Alice to Bob and the fixed-direction polytope with communication from Bob to Alice.

\subsection{Computational Procedure}

Enumerating all the vertices can easily be accomplished using computer software such as MATLAB.  What interests us are the corresponding facet inequalities of the polytope, which represent the ``Bell Inequalities'' for the particular model.  Converting the vertex characterization into the facet characterization of a polytope is known as the \textit{hull problem}, and in general it is an NP-Complete problem \cite{Pitowsky-1991a}.  The computational task is made even more laborious due to the exponential growth in the communication complexity resulting from the asymptotic behavior of $\genfrac\{\}{0pt}{}{M_A}{2^r}$.  

For $M_A=3$, $M_B=2$, and $r=1$, we were able to complete the calculation.  After using MATLAB to enumerate all the polytope vertices, the main tool used was the freely available linear optimization program called PORTA \footnote{\url{http://typo.zib.de/opt-long_projects/Software/Porta/}} that enabled us to switch between polytope representations. Lastly, we wrote a Python program to convert between locally equivalent inequalities and remove those equivalent inequalities.

\section{Facet Inequalities for $M_A=3$, $M_B=2$}

\label{Sect:results}

\subsection{Fixed-Direction Communication}

We now list the Bell Inequalities for the Alice $\to$ Bob fixed-direction communication model.  In what follows, we will represent inequalities by giving the coefficients of the 15 free variables.  The coefficients will be arranged in a chart as: 
\[
\begin{BMAT}{|ccc|}{c|cc|cc}
q_A(0|0)   & q_A(0|1) & q_A(0|2) \\
p(00|00) & p(00|10) & p(00|20) \\
p(00|01) & p(00|11) & p(00|21) \\
p(10|00) & p(10|10) & p(10|20) \\
p(10|01) & p(10|11) & p(10|21)
\end{BMAT}\leq \Gamma.
\]
Here, we multiply whatever numbers appear in the chart by the corresponding probability, and the total sum must be less than $\Gamma$.  For instance, the box
\[
\begin{BMAT}{|ccc|}{c|cc|cc}
-1& 0  & 0 \\
0 & -2 & 3 \\
0 & 4  & 0 \\
0 & 0 & 5 \\
-6 & 0 & 0
\end{BMAT}\leq 7.
\]
means that
\begin{align}
-q_A(0|0)-2p(00|10) + 3p(00|20)+ 4p(00|11)&\notag\\
+5p(10|20)-6p(10|01)&\leq 7.\notag
\end{align}

In addition to the non-negativity constraint and the condition that $q_A(0|i)\geq p(00|ij)$, we obtain the following eight facet inequalities for the case Alice sends 1 bit of CC to Bob.

\begin{align*}
\textbf{I.}\quad &\begin{BMAT}{|ccc|}{c|cc|cc}
-1& 0  & 0 \\
0 & -1 & 1 \\
0 & 1  & 0 \\
-1 & 0 & 1 \\
-1 & 0 & 0
\end{BMAT}\leq 1.
&\textbf{II.}\quad
\begin{BMAT}{|ccc|}{c|cc|cc}
0  & 0  & 0 \\
-1 & -1 & 1 \\
-1 &  1 & 0 \\
-1 & 0  & 1 \\
-1 & 0  & 0
\end{BMAT} \leq 1.
\end{align*}

\begin{align*}
\textbf{III.}\quad &
\begin{BMAT}{|ccc|}{c|cc|cc}
-1  & -1 & 0 \\
0   & 1 & -1 \\
0  & 1 & 1 \\
-1 & 0 & 1 \\
-1 & 0 & -1
\end{BMAT} \leq 1.
&\textbf{IV.}\quad
\begin{BMAT}{|ccc|}{c|cc|cc}
0  & 0  & 0 \\
-1 & -1 & 1 \\
-1 &  1 & 0 \\
-1 & -1 & 1 \\
-1 & 1  & 0
\end{BMAT} \leq 1.
\end{align*}

\begin{align*}
\textbf{V.}\quad &
\begin{BMAT}{|ccc|}{c|cc|cc}
-1 & 0 & 0 \\
1 & -1 & -1 \\
1 & -1 &  1 \\
0 & -1 & 1 \\
0 & -1 & -1
\end{BMAT} \leq 1.
&\textbf{VI.}\quad
\begin{BMAT}{|ccc|}{c|cc|cc}
-3   & 0 & 0 \\
2 & -2 & 0 \\
2 & 1 & -1 \\
-1 & -1 & 1\\
-1 & 0  & -2
\end{BMAT} \leq 1.
\end{align*}

\begin{align*}
\textbf{VII.}\quad &
\begin{BMAT}{|ccc|}{c|cc|cc}
-3   & 0 & 0 \\
2 & -2 & 0 \\
2 & 1 & -1 \\
-1 & -2 & 1\\
-1 & 1  & -2
\end{BMAT} \leq 1.
&\textbf{VIII.}\quad
\begin{BMAT}{|ccc|}{c|cc|cc}
-3   & 0 & 0 \\
2 & -2 & 1 \\
2 & 1 & -2 \\
-1 & -2 & 1\\
-1 & 1  & -2
\end{BMAT} \leq 1.
\end{align*}

On the other hand, if the fixed communication direction is Bob $\to$ Alice, then the allowed correlations are precisely those that satisfy the no-signaling from Alice to Bob (i.e. Eq. \eqref{Eq:No-Signal-Alice}).  A protocol for simulating any such $p(ab|ij)$ is reviewed in Sect. \ref{Sect:LowerBounds}.

\subsection{Bi-directional Communication}

We have computed a total of 143 inequivalent facet inequalities for bi-directional CC-assisted correlations.  These are too numerous to present here, and we just make a few simple remarks.  First, both the Alice $\to$ Bob and Bob $\to$ Alice fixed-direction polytopes are contained in the bi-directional polytope.  On the other hand, there exist certain distributions that violate at least one of the 143 inequalities.  For instance, like the distribution given in Ref. \cite{Bacon-2003a}, any distribution of the form $p_{ab|ij}=a(j)b(i)$ lies outside the random-direction polytope, where both $a$ and $b$ are non-constant functions of $j$ and $i$ respectively.  Intuitively, these distributions correspond to the scenario in which both Alice and Bob's output depend on the other's input.  This cannot be simulated with only one bit of communication shared between the duo.   

\section{Communication Cost for Simulating No-Signaling}
\label{Sect:LowerBounds}

In this section we quantify the CC cost for simulating a general no-signaling correlation.  Specifically, we show that the simulation protocol given by Bacon and Toner is optimal.  For $M_A\geq M_B$, their protocol uses $\lceil\log_2 M_B\rceil$ bits of CC and is given as follows \cite{Bacon-2003a}.  For each $j\in\{0,1,...,M_B-1\}$, Alice and Bob share a random variable ranging over $\{0,1,...,K_B-1\}$ with distribution given by $\{q_B(b|j)\}_{b=0}^{K_B-1}$.  Bob sends his input $j$ to Alice and they both consult their share correlations to obtain output $b$.  Then for each of Alice's input $i\in\{0,1,...,M_A-1\}$, she samples from $\{0,1,...,K_A-1\}$ with distribution $\{p(ab|ij)\}_{a=0}^{K_A-1}$ and outputs $a$.

We now show that this protocol is, in fact, optimal.
\begin{theorem}
\label{Thm:main}
Consider boxes with $M_A/K_A$ inputs/outputs for Alice and $M_B/K_B$ inputs/outputs for Bob, with $M_A\geq M_B$.  Then simulating all no-signaling correlations with local shared randomness requires at least $\lceil\log_2 M_B\rceil$ bits of CC.
\end{theorem}
\begin{proof}
For a given $M_A$ and $M_B$, we will construct a non-signaling correlation, $\hat{p}(ab|ij)$, that cannot be simulated with fewer than $\lceil\log_2 M_B\rceil$ bits of CC.  In this distribution, only outputs $0$ and $1$ have nonzero probabilities.  For $i=0$, outputs 0 and 1 are perfectly correlated: $\hat{p}(00|ij)=\hat{p}(11|ij)=1/2$.  For $0< i\leq M_B-1$, the distributions are (i) perfectly correlated for outputs 0 and 1 whenever $i\not=j$, and perfectly anti-correlated for outputs 0 and 1 whenever $i=j$: $\hat{p}(01|ii)=\hat{p}(10|ii)=1/2$.  For $i\geq M_B$, the distributions are deterministic for Alice: $\hat{p}(00|ij)=\hat{p}(01|ij)=1/2$.  It can easily be verified that these correlations satisfy Eqns. \eqref{Eq:No-Signal-Alice} and \eqref{Eq:No-Signal-Bob}.  In fact, this is an extreme point of the binary output no-signaling polytope, as proven by Jones and Masanes \cite{Jones-2005a}.

Suppose that $r=\lceil\log_2 M_B\rceil-1$ bits of CC suffice to simulate the given distribution.  As described in Sect. \ref{Sect:CC}, a general CC-assisted strategy can be decomposed into a convex combination of deterministic strategies.  Let $\mathcal{S}$ be a variable for the various strategies, each occurring with probability $p(\mathcal{S})$.  The total correlations are then given by $\hat{p}(ab|ij)=\sum_{\mathcal{S}}p({\mathcal{S}})p(ab|ij\mathcal{S})$.  For any particular strategy $\overline{\mathcal{S}}$, Alice sends Bob $\kappa(i)\in\{0,1,\ldots,2^{s}-1\}$ and Bob sends Alice $\sigma(j)\in\{0,1,\ldots,2^{r-s}-1\}$.   
\begin{proposition}
There exists three distinct values $t_0,t_1,t_2\in\{0,1,...,M_B-1\}$ such that $\kappa(t_0)=\kappa(t_2)$ and $\sigma(t_1)=\sigma(t_2)$.
\end{proposition}
\begin{proof}
For $x\in\{0,1,\ldots,2^{s}-1\}$, let $n_{max}=\max_x|\kappa^{-1}(x)|$ and take $x_0$ such that $|\kappa^{-1}(x_0)|=n_{max}$.  We have $n_{max}\geq\frac{M_B}{2^s}>2^{r-s}$.  If $r=s$, choose $t_0,t_2\in\kappa^{-1}(x_0)$ and $t_1$ to be any other nonnegative number $\leq M_B-1$.  If $r>s$, since $\sigma$ takes on $2^{r-s}$ different values, it follows that there must exist at least one distinct pair $t_1,t_2\in\kappa^{-1}(x_0)$ such that $\sigma(t_1)=\sigma(t_2)$.  The cardinality bound $|\kappa^{-1}(x_0)|>2$ means a third distinct $t_0$ can be found in $\kappa^{-1}(x_0)$.
\end{proof}
Let $t_0,t_1,t_2$ be three values described in the proposition, and consider the four sets of conditional probabilities $p(ab|i j\overline{\mathcal{S}})$, with $i\in\{t_0,t_2\}$ and $j\in\{t_1,t_2\}$.  The equality $\hat{p}(ab|ij)=\sum_{\mathcal{S}}p(\mathcal{S})p(ab|ij\mathcal{S})$ implies that $p(ab|ij\overline{\mathcal{S}})$ is zero whenever $\hat{p}(ab|ij)$ is zero.  Hence, 
\begin{align}
\label{Eq:GenProbs}
0&=p(01|t_0 t_1\overline{\mathcal{S}})=p(10|t_0 t_1\overline{\mathcal{S}})\notag\\
0&=p(01|t_0 t_2\overline{\mathcal{S}})=p(10|t_0 t_2\overline{\mathcal{S}})\notag\\
0&=p(01|t_2 t_1\overline{\mathcal{S}})=p(10|t_2 t_1\overline{\mathcal{S}})\notag\\
0&=p(00|t_2 t_2\overline{\mathcal{S}})=p(11|t_2 t_2\overline{\mathcal{S}}).
\end{align}
The probabilities satisfy 
\[p(ab|ij\overline{\mathcal{S}})=q_A(a|i\sigma(j)\overline{\mathcal{S}})q_B(b|j\kappa(i)\overline{\mathcal{S}})\]
 with $1=\sum_{a=0}^1 q_A(a|i\sigma(j)\overline{\mathcal{S}})=\sum_{b=0}^1q_B(b|j\kappa(i)\overline{\mathcal{S}})$.  By combining this with the equalities $\kappa(t_0)=\kappa(t_2)$ and $\sigma(t_1)=\sigma(t_2)$, it is straightforward to see that Eqns. \eqref{Eq:GenProbs} cannot be simultaneously satisfied.
\end{proof}

\section{Conclusion}

\label{Sect:conclusion}

In this paper we have made partial progress in understanding how classical communication functions as a resource in generating non-local correlations.  We have completely characterized the correlation polytope for $M_A=3$ and $M_B=K_A=K_B=2$.  This is done for both fixed-direction and bi-directional communication.  We then proved optimality in communication cost for simulating no-signaling correlations; regardless of the number of measurement outcomes, at least $\min\{\lceil\log_2 M_A\rceil,\lceil\log_2 M_B\rceil\}$ bits of CC are required to produce all no-signaling correlations.

Concerning the question of simulating quantum correlations, one bit of CC is sufficient so long as $\min\{M_A,M_B\}=2$.  Hence, all quantum strategies will satisfy the inequalities presented in Section \ref{Sect:results}.  The next obvious scenario to consider is $M_A=M_B=3$ and $K_A=K_B=2$.  Here, extensive work has been conducted to understand the local \cite{Pitowski-2001a, Collins-2004a}, quantum \cite{Navascues-2007a, Pal-2010a}, and more general non-local correlations \cite{Brunner-2005a}.  While our Theorem \ref{Thm:main} shows that one bit of CC is insufficient for simulating all non-signaling correlations, it is an important problem to understand whether all quantum correlations can nevertheless be simulated.

\begin{acknowledgments}
We would like to thank Benjamin Fortescue and Min-Hsiu Hsieh for helpful discussions during work on this project.  K.M. was supported under the Research Experience for Undergraduates (REU) grant NSF DMR 1157058.
\end{acknowledgments}

\bibliography{QuantumBib}

\end{document}